\documentclass[letterpaper,twocolumn,10pt]{article}
\usepackage{usenix2019_v3}

\usepackage{amsmath}
\usepackage{amsthm}
\usepackage{amssymb}
\usepackage{bm}
\usepackage{bookmark}
\usepackage{dashbox}
\usepackage[T1]{fontenc}
\usepackage{graphicx}
\usepackage[utf8]{inputenc}
\usepackage{mathtools}
\usepackage{mathrsfs}
\usepackage{marvosym}
\usepackage{newunicodechar}
\usepackage{tabu}
\usepackage{thm-restate}
\usepackage{xspace}
\usepackage{xcolor}
\usepackage{stmaryrd}
\usepackage{eufrak}
\usepackage{amsmath}
\usepackage{colortbl}
\usepackage{subcaption}
\usepackage{enumitem}
 
\newcommand*{\smparagraph}[1]{\vspace{0.2cm}\noindent{{\bf #1}: }}
\newcommand*{\itparagraph}[1]{\vspace{0.2cm}\noindent{{\it #1}: }}

\makeatletter
\renewcommand{\boxed}[1]{\text{\fboxsep=.2em\fbox{\m@th$\displaystyle#1$}}}
\makeatother

\newcommand*\adv{\mathcal{A}}

\newcommand*\manu{\mathcal{M}}

\newtheorem{assumption}{Assumption}
\theoremstyle{definition}
\newtheorem{definition}{Definition}

\newcommand*\ids{I}
\newcommand*\pcorrupt{C}
\newcommand*\secprob{p}

\newcommand{\define}{:=} 

\newcommand{\numcusto}{n} % number of custodians in the system
\newcommand{\numdel}{m} % number of delegates
\newcommand*\numd{N}  % number of devices
\newcommand{\threshdel}{t} % threshold of delegates (out of \numdel) needed to unlock device

\newcommand*\scheme{\textsc{JJE}\xspace}
\newcommand{\tok}{\phi} % unlock token
\newcommand{\nonce}{\alpha} % nonce issued by device
\newcommand*\device{\mathcal{D}} % delegate device 
\newcommand*\locdevice{\boxed{\device}} % the locked device that LE needs access to
\newcommand*\custo{\mathcal{C}} % custodian
\newcommand*\deviceid{\mathsf{IMEI}} % imei of a device
\newcommand*\enc[1]{\boldsymbol{[}{#1}\boldsymbol{]_\text{\lock}}} % encryption of a value
\newcommand*\delid[1]{\mathsf{DelegateID}_{#1}} % delegate identifier
\newcommand*\keydb{\mathcal{L}} % database of keys
\newcommand*\sig{\sigma} % signature
\newcommand*\epoch{\tau} % current epoch
\newcommand*\chal{\texttt{chal}} % challenge issued by device to LE
\newcommand*\pk{\mathsf{pk}} % public encryption key of an entity
\newcommand*\vk{\mathsf{vk}} % public verification key of an entity
\newcommand*\sk{\mathsf{sk}} % secret key 
\newcommand*\vkm{\vk_{\manu}} % public verification key of the manufacturer
\newcommand*\skm{\sk_{\manu}} % secret key of manufacturer
\newcommand*\vkj{\vk_{\custo}} % public verification key of custodians
\newcommand*\pkj{\pk_{\custo}}  % public encryption key of custodians
\newcommand*\skj{\sk_{\custo}}  % secret key of custodians
\newcommand*\vkd{\vk_{\device}} % verification key of a delegate device
\newcommand*\skd{\sk_{\device}} % secret key of a delegate device
\newcommand*\delselmal{\Gamma} % selection of delegates given to device by LE
\newcommand*\delsel{\Sigma} % challenge selection of delegates issued by locked device
\newcommand*\keydbheader{H}  % header metadata of keydb
\newcommand*\header{(\MerkleDigest(\keydb)||\numd||\epoch)}
\newcommand*\headerepoch{(\MerkleDigest(\keydb_\epoch)||\numd_\epoch||\epoch)}

\newcommand*\CustodianConsensus{\textsf{CustodianConsensus}}
\newcommand*\SystemSetup{\textsf{SystemSetup}}
\newcommand*\DeviceSetup{\textsf{DeviceSetup}}
\newcommand*\DelegateRegister{\textsf{DelegateRegister}}
\newcommand*\SelectDelegation{\textsf{SelectDelegation}}
\newcommand*\RevealChallenge{\textsf{RevealChallenge}}
\newcommand*\DeviceUnlock{\textsf{RevealUnlockToken}}
\newcommand*\DelegateSignRequest{\textsf{SignRequest}}
\newcommand*\MerkleVerify{\texttt{MerkleVer}}
\newcommand*\MerkleDigest{\texttt{MerkleHash}}
\newcommand*\MerklePath{\texttt{MerklePath}}
\newcommand*\MerkleProof{\Pi}

\hypersetup{linkcolor=black}

\begin{document}
 
%don't want date printed
\date{}

\title{\Large \bf Judge, Jury \& Encryptioner:\\
  Exceptional Device Access with a Social Cost}

%for single author (just remove % characters)
\author{
{\rm Sacha Servan-Schreiber}\thanks{Both authors contributed equally to this work.}~~\thanks{3s@mit.edu}~~$^{(\text{\Letter})}$\\
{\rm MIT CSAIL}\\
\and
{\rm Archer Wheeler}\footnotemark[1]~~\thanks{archer\_wheeler@brown.edu}\\\
{\rm Brown University}\\
} % end author

\maketitle

\begin{abstract}
% !TEX root = main.tex

We present Judge, Jury and Encryptioner (\scheme) an exceptional access scheme for unlocking devices that does not give unilateral power to any single authority. 
\scheme achieves this by placing final approval to unlock a device in the hands of peer devices.
\scheme distributes maintenance of the protocol across a network of ``custodians'' such as courts, government agencies, civil rights watchdogs, and academic institutions.
Unlock requests, however, can only be approved by a \emph{randomly selected} set of recently active peer devices that must be \emph{physically located} by law enforcement in order to gain access to the locked device. 
This requires that law enforcement expend both human and monetary resources and pay a ``social cost'' in order to find and request the participation of random device owners in the unlock process.

\smallskip

Compared to other proposed exceptional access schemes, we believe that \scheme mitigates the risk of mass surveillance, law enforcement abuse, and vulnerability to unlawful attackers. 
While we propose a concrete construction, our primary goal with \scheme is to spur discussion on \emph{ethical} exceptional access schemes that balance privacy of individuals and the desires of law enforcement. 
\scheme transparently reveals the use of exceptional access to the public and enforces a fixed social cost that, we believe, can be an effective deterrent to mass surveillance and abuse.

\end{abstract}

% !TEX root = main.tex

\section{Introduction} \label{sec:intro}
 On December 7th, 2018, Australia passed a series of controversial and contested laws 
 designed to compel technology companies to provide law enforcement and government 
 agencies access to all encrypted communications~\cite{bbc2018australia, newman2018australia}. 
 In September of 2019, the U.K. government ruled that Facebook and WhatsApp are required to 
 share private messages with the police to assist in criminal investigations~\cite{castro2019}.
 More recently, the Attorney General of the United States called on Apple to unlock two iphones owned by a terrorist with no clear plan for how to do so without jeopardizing the security and privacy of law abiding citizens~\cite{benner_2020}.
 
These events raise important security, political, and ethical questions for both the åcryptography and technology policy communities. 
This recent legislation in both Australia and the U.K. is of great interest to the Five Eyes intelligence compact (U.S., U.K., Australia, Canada \& New Zealand). 
In 2018, the Five Eyes issued a joint statement advocating for lawful electronic interception~\cite{fiveeyes}. 
These recent developments raise the fear that other countries, specifically the United States, will enact similar laws to that of Australia and the U.K.~\cite{wiredglobalpriv}.

The advent of encryption-by-default for many smartphones and instant messaging applications has left law enforcement and intelligence agencies worried that many of their surveillance tools, purportedly used for national security, could ``go dark''~\cite{fiveeyes}. 
As several recent events have demonstrated, law enforcement and government officials firmly believe that they are entitled to ``exceptional access'' to encrypted data and devices, sparking significant pushback by both the private sector and the academic cryptography community~\cite{fiveeyes, weitzner2018encryption}. 

There is an understanding among the security and privacy experts that building ``backdoors'' into devices fundamentally compromises security~\cite{abelson2015keys, buchanan2016cybersecurity}. 
Unfortunately, the technical nuances of this reasoning is often lost in the policy making process. 
Moreover, the vocal opposition to giving law enforcement access to devices has not been successful in stopping countries from pushing for more legislation surrounding intercept technologies~\cite{wiredglobalpriv, bbc2018australia, newman2018australia}. 
It is therefore imperative to understand the consequences and discuss practical solutions to this problem. 
{\bf Our goal is not to advocate for the deployment of exceptional access but rather to instigate discussion into the trade offs needed to build ethical exceptional access schemes.}

As recently pointed out by Ray Ozzie~\cite{ozzieclear}, it is the responsibility of the cryptography community to ensure that backdoors are not implemented in a manner which could have serious long-term security implications. 
Our goal is to address the concerns of law enforcement at face value while still constructing a system which mitigates the concerns of abuse posed by the cryptographic and policy community. 
We do not expect our proposal will be preferable to law enforcement compared to existing methods of ``backdoors'' or key escrow. Rather, we hope to show that ethical solutions are feasible, and we hope to provide a counter argument to the reasoning that \emph{carte blanche} government backdoors are necessary to maintain effective national security.  
 
In this work, we present an exceptional access scheme, \scheme, that decentralizes the trust and unlock process without placing unilateral power in the government’s hands. 
From a high level, our scheme works as follows: At set up, each device secretly and blindly chooses a number of other anonymous devices to act as its ``access delegates.'' 
Crucially, the device chooses these delegates at random with \emph{no knowledge of who they are selecting} and this selection does not leave the device. 
It then builds a unique exceptional access challenge which can be turned into an unlock token by a subset of the chosen delegates.

To use the exceptional access scheme, law enforcement needs approval from a set of ``custodians'' in order to identify and \emph{physically locate} the delegates. 
Each delegate device provides a necessary signature through a secure hardware protocol that does not compromise the security or privacy of the delegate. 
Once a sufficient number of signatures are obtained, they can be used to unlock the locked device. 
Finding and accessing all of the needed delegate devices requires broad manpower and access to intelligence that only the government has. 
This aspect of \scheme is intangible and difficult to steal even by a determined adversary, making \scheme robust against malicious foreign state actors or other well funded attackers. 

\subsection{Ethical Implications}\label{ethics}
The strong opinion of the cryptography community is that building backdoor technology is inherently misguided as it cannot be done while maintaining the security of cryptography~\cite{levycracking}. 
However, the vocal opposition of the community has not caused the problem to go away~\cite{bbc2018australia, newman2018australia, wiredglobalpriv}. 
It is apparent that there is a disconnect between researchers who deeply understand the problem and legal experts who, perhaps naively,  seek  \emph{ a solution that simply works}. 
What governments want is a slight weakening in individual privacy for the gain of security as a nation~\cite{weitzner2018encryption, edgar2017beyond}. 
On the other hand, many researchers believe that any such ``slight'' weakening exposes the real possibility 
of massive security vulnerabilities and provides governments with yet another tool for mass surveillance~\cite{abelson2015keys}. 

It is tempting for members of the academic cryptography community to answer questions seeking secure backdoors with an unequivocal, ``not possible.'' 
In fact, the technical argument for why this is impossible is a tautology: if you provide access to someone else, then your system is insecure because someone else, possibly less trustworthy, now has access too. 
History has shown that this is not just a problem in theory. In fact, any centralized database, no matter how practically secure, is susceptible to attack~\cite{nakashima2013chinese}. 
While this fact is likely obvious to the technically savvy reader, it is not so inherently obvious to a lay person or to a politician~\cite{weitzner2018encryption} making the privacy community come across as recalcitrant on the subject. 

The situation is complicated by the fact that the political playbook used to argue for the existence of exceptional access plays to fears and emotions. 
A common argument used by policy makers advocating \emph{against encryption} goes as follows: 
Imagine there is a terrorist who will imminently attack unless we can unlock this phone, on the spot, and discover their plan~\cite{wjs2016applevsfbi}. 
By focusing on such an extreme scenario, advocates against strong encryption argue it is reasonable to weaken security \emph{for all devices}.

However, currently deployed and proposed systems allow governments to use the same access system for mass surveillance, or to weaken the legal rights of individuals simply accused of low level crimes~\cite{weitzner2018encryption}.

Ultimately, the legal line between security and privacy will be determined not by cryptographers, but by policy experts, the industry, and consumers. 
However, we believe that it is of crucial importance that the research community engages in a discussion regarding potential solutions in order to attempt to bridge a gap, and that it does so judiciously. 
Our goal is to instigate research of \emph{ethical} exceptional access that is fully transparent to the public and incurs high social costs making it impractical as a tool for mass surveillance. 
We believe that any proposed exceptional access scheme should be  designed, from the start, to be truly exceptional, i.e., they should only be useful in the most extreme of cases.
For a more in-depth discussion on the ethical implications of exceptional access schemes we refer the reader to~\cite{wright2018crypto}.

\subsection{Authors' Opinions}
Considering the controversy surrounding this area of research, we believe it is important to stress that we, the authors, are not advocating legislation mandating exceptional access schemes. 
We firmly believe that attempts to outlaw strong encryption are misguided and dangerous. 
However, we also believe that the cryptographic research community needs to take this problem seriously in order to avoid rushed legislation that gives governments unilateral power to unlock devices. 
While we are aware that credible exceptional access research could be used as a fig leaf to justify backdoors that are secure ``in name only,'' recent events show governments are willing to mandate them regardless of academic research~\cite{wiredglobalpriv, bbc2018australia, newman2018australia}. 
We hope \scheme can help facilitate a discussion into ethical exceptional access and to show that there are more secure and more ethical alternatives to secretive ``backdoors''. 
We believe that an honest discussion will  strengthen the hand of privacy advocates.

\subsection{Our Contributions}
To the best of our knowledge, \scheme is the first exceptional access proposal which is fully decentralized, ensures legal authorities are engaged in the unlock process, and incurs a high operational cost \emph{per device} when used by law enforcement. 
In addition, using \scheme requires a fixed, unavoidable social cost. Our \emph{technical} contributions are twofold. 

\begin{itemize}
	\item We present the first decentralized system for providing exceptional access where accountability is enforced and social costs are imposed on a per device basis;
	\item We extensively analyze our system under a realistic threat model to ensure it does not lend itself as a tool for mass-surveillance and other forms of abuse.
\end{itemize}

\noindent While we believe that \scheme is a novel system for exceptional access which improves upon the existing proposals, there are several limitations to \scheme which we believe are import to address upfront. 

\subsection{Limitations}
\scheme is one of few academic works rigorously addressing the problem of exceptional access to locked devices~\cite{savage2018lawful, wright2018crypto}.
As such, \scheme is necessarily a starting point for future work in this area intersecting cryptography, systems, and policy research.
\scheme should be viewed as a stepping stone and instigator for (academic) discussion rather than a perfect ``deployable tomorrow'' solution to the multifaceted and complex problem of exceptional access. 

\begin{itemize}
	\item \scheme is designed to be used in a country where there exists some level of trust in independent legal or non-government institutions. This may be infeasible in a totalitarian government.

	\item \scheme requires a secure enclave processor. However, without such hardware it is likely a device could be easily unlocked by a sophisticated adversary~\cite{wjs2016applevsfbi, costan2017secure, costan2016sanctum, sau2017survey, levycracking}
	
	\item \scheme could present novel legal questions about international jurisdiction of locked devices. We discuss how \scheme can handle these issues in Section \ref{sec:jurisdiction}. However, this issue fundamentally  arises from denying governments unilateral authority to unlock devices.
\end{itemize}

\noindent We provide additional discussion on limitations in Section~\ref{sec:limitations}.

% !TEX root = main.tex

\section{Related Work}
\label{sec:related}
\smparagraph{Exceptional Access \& Key Escrow}
There have been only a few recent works related to exceptional access. 
The existing work has either focused on providing law enforcement with access to encrypted communications or proposing systems for lawful device unlocking under certain assumptions. 
Older work (circa 1995), has been primarily interested in the problem of \emph{key escrow} rather than access to locked devices. 
This includes work by Bellare and Goldwasser~\cite{bellare1997verifiable, bellare1996encapsulated,goldwasser1998time} where the authors describe several different techniques for key escrow systems, where the trust placed in law enforcement is limited. 
At a high level, they proposes different means of providing a partial secret key to a trusted 3rd party such that after the partial key is recovered, obtaining the full secret key requires computational (and monetary) resources with the hope that such cost would limit mass surveillance. 

It is only recently (circa 2018), following the heavily publicized Apple vs. FBI debacle~\cite{wjs2016applevsfbi}, that the community had revived interest in this esearch area. 
In 2018, Wright and Varia~\cite{wright2018crypto} propose a scheme for giving law enforcement access to encrypted data using time-lock puzzles. 
Their approach, in theory, requires law enforcement to expend significant \emph{monetary} resources to break encryption by requiring abundant computing power and upfront costs for pre-computation. 
These costs, they claim, can be tuned so that each intercepted \emph{message} costs between \$1000 and \$1M to crack. Their approach does not involve any technical means for enforcing transparency or incorporating social costs. 

In the same year, Ozzie~\cite{ozzieclear} proposed a system for unlocking \emph{devices} by placing full trust in the device manufacturer. 
While the proposal lacks much technical detail and analysis, the high-level idea is to have the manufacturer act as a trusted key escrow, acting as an arbiter for the unlock process. 
In the proposed scheme, each device encrypts its own ``master key'' with the manufacturer's public key which can then be decrypted by the manufacturer in cooperation with law enforcement to provide access to the device. 

Savage~\cite{savage2018lawful} expands and formalizes several of Ozzie's ideas to propose a similar system for lawful device unlocking by placing power and trust in the manufacturer's hands. 
The scheme also sets in place several safeguards for preventing mass unlock capabilities by law enforcement, ensuring post unlock transparency, and time vaulting. 
From our perspective, \cite{savage2018lawful} is most related to \scheme, but differs from our approach in that there are no baked-in guarantees of (public) transparency and social cost. 

\smparagraph{Surveillance \& Cryptography}
Orthogonal to \scheme is the topic of \emph{surveillance}. 
Using cryptography for creating accountability in surveillance has been studied in several prior works. 
Kroll~\emph{et al.}~\cite{kroll2014secure} propose a mechanism for providing accountability in the court system using secure multi-party computations to execute warrant requests in a private yet auditable manner. 

Similarly, Bates~\emph{et al.}~\cite{bates2015accountable} propose a system for accountable wiretapping technology with the goal of keeping law enforcement activity in the public record.  
This line of work was followed up more recently by Frankle~\emph{et al.}~\cite{frankle2018practical} where the authors discuss ways in which secret processes can be made publicly auditable without compromising secrecy of law enforcement investigations. 
The system uses a combination of zero-knowledge proofs, a tamper-proof ledger, and multi-party computation to achieve that goal. 

Goldwasser~\emph{et al.}~\cite{goldwasser2017public} describe a cryptographic system for allowing accountability of secret laws related to the Foreign Intelligence Surveillance Act (FISA). 
Their solution uses zero-knowledge proofs to provide public auditing of (secret) data collection and surveillance procedures. 

These works (and many others) provide interesting avenues for exploring extensions to our system, especially in regards to providing different levels of accountability and privacy in exceptional access mechanisms but we do not consider them to be part of the \emph{exceptional access} literature.

% !TEX root = main.tex

%%%%%%%%%%%%%%%%%%%%%%%%%%%%%%%%%%%%%%%%%
\begin{table*}[t]
\centering
\tabulinesep=1.2mm
\begin{tabu} to \textwidth {| X[l] | X[l] | X[l] | X[l]| X[l] |}
\hline
 \textbf{Authorization} 
 & \textbf{Non-scalability} 
 & \textbf{Physical Access} 
 & \textbf{Transparency} 
 & \textbf{Social Cost} \\
 \hline
 Identity of the selected delegation can only be revealed by the custodians. 
 & Finite resources must be expended by law enforcement to unlock each device.
 & Physical possession of a device is required to initiate the unlock process.
 & All unlock request require participation of multiple custodians and peer devices.
 & Exceptional access must incur a social cost on a per-device basis. \\
 
 \hline
\end{tabu}

    \caption{Summary of goals for \scheme.}
    \label{tab:goalsandmethods}
\end{table*}
%%%%%%%%%%%%%%%%%%%%%%%%%%%%%%%%%%%%%%%%%

\section{Overview}
\label{sec:overview}
\subsection{Goals}
The primary goal of any system designed to give legal access to a third party is to avoid abuse and ensure transparent legal compliance. 
The concern of backdoors is their potential for fueling mass surveillance practices by a malicious government agency, rogue employees, and hackers. 
Hence, the foremost priority of designing a lawful access scheme is to bake in limits on the scheme's usage at the design level. 
Ozzie's proposal~\cite{ozzieclear} fails to identify this particular risk, instead relying on the private sector to act as a restrainer. 
Stefan Savage, in his technical discussion~\cite{savage2018lawful}, identifies the crucial need for imposing clear \emph{technical} barriers to government mass surveillance. 
However, the proposed system still requires the private sector to keep the government in check, albeit to a lesser extent by imposing time-lock mechanisms to limit covert unlocking capability. 
We establish the following desirable properties which we believe a lawful device access scheme should achieve.
We provide a summary of the desired properties in Table~\ref{tab:goalsandmethods}.

\smparagraph{Authorization}
Law enforcement must provide proof-of-authorization to unlock a device (e.g., a warrant). 
This is an imperative first step to avoiding abuse of power; however, it is not sufficient on its own as seen from numerous historical precedent~\cite{edgar2017beyond,nsa2013primary}.

\smparagraph{Non-scalability}
Individual device access is infeasible to scale to mass surveillance operations. 
We can do this by requiring law enforcement to expend a finite resource (e.g., physical or monetary) to obtain access. 
The access mechanism should incur a certain \emph{physical} cost to law enforcement and could take the form of officers on the ground or task force organization on a \emph{per device} basis. 
Other forms of physical costs could include \emph{computational} or \emph{monetary} costs, as proposed in~\cite{wright2018crypto}. 

\smparagraph{Social Costs}
While non-scalability via physical costs is important, it is not sufficient on its own. 
A well funded adversary (e.g., a government agency) can adapt by increasing its budget or creating specialized task forces. 
Thus, we additionally require \emph{social costs} --- social and political costs levied through transparency which act as deterrent to mass surveillance. 
With both physical and social costs imposed on the unlock mechanism, the barriers to mass surveillance can be fine-tuned such that governments are unable to orchestrate device access operations \emph{en masse}.

\smparagraph{Physical Access} A related goal which serves the purpose of guaranteeing non-scalability is ensuring that no device unlock procedure can be instigated remotely. 
In other words, physical possession of the device is required to unlock it. 
We therefore desire that law enforcement posses the device prior to being able to issue \emph{any aspect} of the unlock request. 
We propose cryptographic methods that are used in conjunction with existing hardware to achieve this requirement. 

\smparagraph{Transparency}
Finally, it is crucial to maintain transparency of the unlock procedure, i.e., there must be some level of public awareness to the use of the exceptional access scheme in order to avoid covert spying. 
Transparency serves as a deterrent to covert mass surveillance since government agencies can be held accountable for their actions.

\subsection{Why Devices and not Encryption}
A crucial distinction is that \scheme does not affect security of \emph{encryption} itself. 
Much work has been dedicated to arguing why building backdoors into encryption schemes is simply not feasible without compromising security at some level~\cite{abelson2015keys, weitzner2018encryption}. 
This work does not attempt to refute those claims.

This paper addresses the more reasonable demand of accessing \emph{personal devices} that are encrypted by default.
We do not make any attempt to weaken encryption or give law enforcement access to encrypted communication. 
However, it is important to note that in many cases, the information stored on personal devices contains most of the data law enforcement may wish to gather on an individual. 
Encrypted and secure communication applications such as Signal, WhatsApp, and iMessage, by default, \emph{are secure solely on the premise that the device on which they are installed is encrypted and locked}. 
Hence, it is reasonable to assume that with access to personal devices, law enforcement can obtain access to otherwise encrypted data as well.

% !TEX root = main.tex

\section{Framework}
\label{sec:design}
We are now ready to describe the \scheme framework. 
We begin by delineating roles and establishing the assumptions we make in the security model.

\subsection*{Roles \& Threat Model}
\scheme has four entities which interact together while following a specified unlock protocol. 
The entities are \emph{Manufacturer}, \emph{Custodians}, \emph{Law Enforcement} and \emph{Access Delegates}. 
In this subsection, we describe the role of each entity and the trust model we assume for each role.

\smparagraph{Manufacturer}
A device manufacturer is responsible for ensuring each device produced is designed with hardware that supports \scheme and isn't compromised at manufacture time. 
In the most basic terms, a manufacturer must be trusted to make devices that have a secure processor that isn't tampered or vulnerable to attack. 

\itparagraph{Threat Model}
We assume that the manufacturer is honest-but-curious, that is, the manufacturer will not \emph{actively} sabotage the unlock mechanism to either allow or prevent law enforcement from accessing a device~\cite{wjs2016applevsfbi}. 
This is a necessary assumption to make given that a fully malicious manufacturer can compromise the hardware, device encryption, and even plant deliberate software bugs. 
However, given sufficient advances in hardware attestation it is possible a manufacture could prove each device follows \scheme and has not been tampered with~\cite{coker2011principles, brickell2004direct, costan2017secure}.
This could allow the manufacture to be reclassified as malicious under our threat model.

\smparagraph{Custodians}
A custodian corresponds to an entity in possession of a (partial) secret key and is fixed 
at system set up time. We imagine custodians as servers belonging and operated by the government itself (e.g., district courts), nonprofit groups, corporations, or academic institutions. 
We require a group of custodians in our system to approve unlock requests issued by government agencies and law enforcement. 

Specifically, custodians are tasked with the role of recognizing legitimate  unlock requests issued by law enforcement. We imagine that organizations acting as custodians have the legal knowledge and funding to validate court issued subpoenas, and for them to not be easily pressured into cooperation. {\it Importantly, custodians have the power to approve device access requests, but not the power to unlock devices. } 

\itparagraph{Threat Model}
We model individual custodians as fully-malicious (i.e., deviating from protocol by collaborating with law enforcement) but assume that at least one custodian is honest-but-curious and follows the protocol. 
Again, this model captures a real-world expectation of trust placed in many legal systems and public institutions such as civil rights groups. It is important to note that \scheme may be less practical in countries without some institutions capable of standing up to government abuse.
We discuss this limitation in Section~\ref{sec:limitations}.

\smparagraph{Delegates}
Delegates are peer devices that collectively have the power to unlock a \emph{single} device at a time. 
A ``delegation'' consists of a randomly chosen set peer devices that act as a type of key escrow. 
A sufficient coalition of devices from a delegation can unlock a device. 
Though we describe the full technical details in subsequent sections, we stress that \scheme ensures that each device has an independently random and anonymous delegation consisting of other devices participating in the protocol. 
Conceptually, a delegation can be seen as a ``jury'' though we do not require delegates to act as ``jurors'' in the legal sense. 
Importantly, we do not expect the owners of delegate devices to decide what is a lawful or unlawful request. The purpose of delegate involvement is to ensure a social cost and enforce transparency of exceptional access.

\itparagraph{Threat Model}
We assume that any device in a selected delegation could be actively malicious. 
In addition, we expect some devices to have limited connectivity, others to be lost, broken, or otherwise not able to cooperate with the protocol. 
Our only requirement is that most devices are not physically compromised by an adversary. 
We will show that, under a reasonable parameter selection, as long as no more than some percentage of all devices (e.g. 25\%) are corrupted, the system is secure from abuse with high probability. We elaborate on the details of this requirement in Section~\ref{sec:security}.

\smparagraph{Law Enforcement}
Law enforcement wish to obtain access to a locked device. To achieve this, law enforcement interacts with custodians to obtain approval, and subsequently interact with delegates within the selected delegation to obtain an unlock token used for unlocking the device. 
In the process, law enforcement will need physical access to all devices in the delegation in order to obtain the necessary signatures to unlock a device. 

\itparagraph{Threat Model}
We model law enforcement as fully malicious. 
Indeed, a large part of the reason for implementing \scheme is to protect law-abiding citizens from corrupted law-enforcement organizations or rogue officers within an agency. 
Moreover, any malicious adversary such as state sponsored hackers would be equivalent under \scheme to rogue law enforcement attempting to (unlawfully) access the device.

\subsection{Assumptions}
\label{subsec:assumptions}
With the roles and threat model established, we now turn to the assumptions that we make in realizing \scheme. 
We make three assumptions which serve as a foundation for constructing the framework. 
We are careful to justify our trust model based on current \emph{real world} security assumptions made in practice. 
Moreover, these assumptions would be required by \emph{any} practical unlock system, not just \scheme.

% assumption 1
\begin{assumption}
\label{assumption:enclave}
	We assume that tamper-reistant, secure enclaves exists and are free of  vulnerabilities.
\end{assumption}

\emph{Justification:}
A device is only as secure as its password. While it is possible to encrypt a device directly, it is impractical for most users to memorize a full encryption key. Some modern devices maintain state of the art security against \emph{physical} attack by using secure enclave co-processors which only run software that has been cryptographically signed by its manufacturer in order to gate access to the true decryption key of the device~\cite{barrett2018titanchip, apple2018t2chip}. 
Indeed, without dedicated tamper resistant hardware, an adversary can be assumed to have the capability of unlocking a device via a brute-force attack, for example.
Because of this, we believe that any exceptional access scheme would be trivially insecure without this requirement.

It is important to note that the term ``secure enclave'' has a somewhat overloaded definition. 
There have been a number of works describing creative attacks on secure processors which allow execution of arbitrary user code, specifically the recent SGX processor developed by 
Intel~\cite{spreitzer2018systematic, mandt2016demystifying, van2018foreshadow}. However, in \scheme it is sufficient for devices to run dedicated software designed by the manufacture at device creation time such as a TPM~\cite{brickell2004direct, apple2018t2chip}. This is the model used by iPhones, 
Macs and some Android phones~\cite{barrett2018titanchip, apple2018t2chip} and is suspected to be difficult to break~\cite{wjs2016applevsfbi}. 
Furthermore, recent work has developed formal guarantees for enclave processors, making them more resilient to attacks~\cite{costan2016sanctum}.

% assumption 2
\begin{assumption}
\label{assumption:software}
    We assume that the only two ways to unlock a device is through the intended mechanism (i.e., by providing a password) and the \scheme exceptional access scheme.
\end{assumption}

\emph{Justification:}
In other words, we do not claim that \scheme fixes \emph{other} methods of breaking into a device.  
Therefore for our security analysis we exclude the possibility of breaking into the device via side-channels and zero-day vulnerabilities. While in the real world this is certainly not assured, 
such methods for accessing devices are tangential to the usefulness of \scheme.

\begin{figure*}[t]
\centering
  \includegraphics[width=\linewidth]{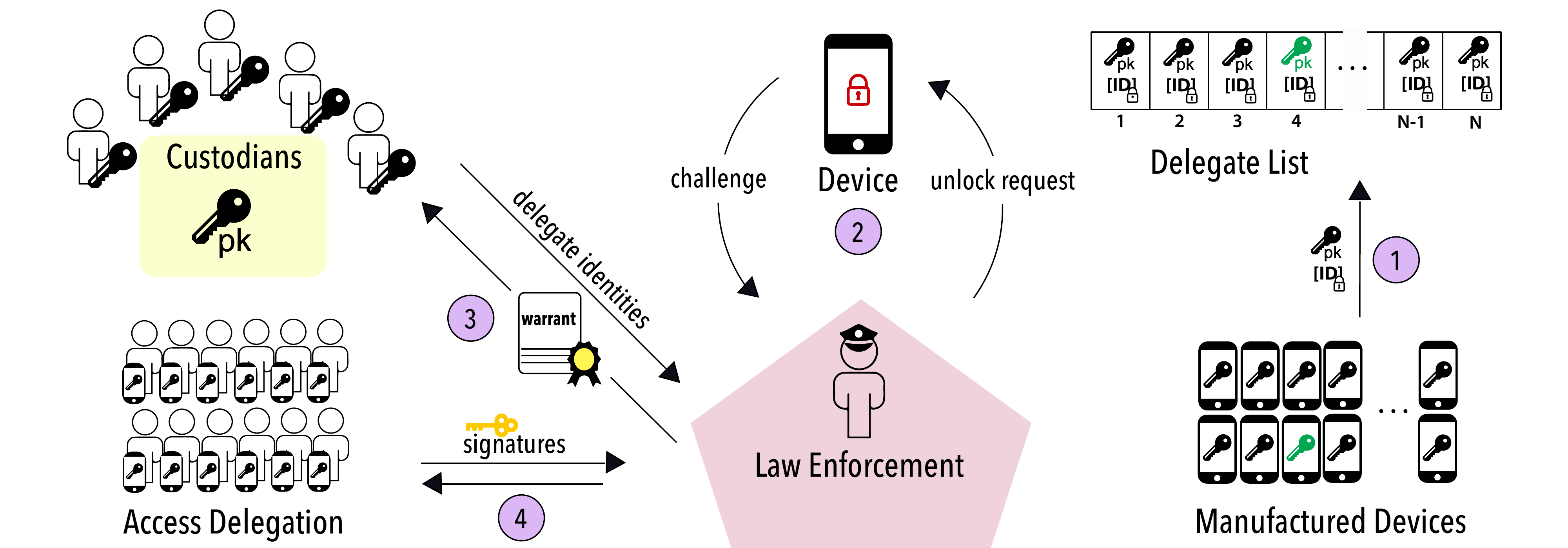}
  \caption{Overview of \scheme protocol and interaction steps necessary to unlock a device. 
  (1) manufactured devices upload a public signature key and encrypted identifier to the delegate list. 
  (2) When an unlock request is issued by law enforcement, the (locked) device selects a random set of indices from the list of delegates which becomes its unlock challenge. 
  (3) Law enforcement requests the identities of the devices located at indices selected by the device. 
  (4) Law enforcement locates the delegates and obtains a set of signatures which are the key to unlocking the device.}
  \label{fig:protocol}
\end{figure*}

% assumption 3
\begin{assumption}
\label{assumption:honesty}
    We assume that only devices honestly following the protocol can be unlocked using the \scheme exceptional access protocol.
\end{assumption}

\emph{Justification:} It is impossible to stop a device from furthering its own security given a savvy and determined user. 
For instance, a user can always encrypt important files using a second layer of encryption which \scheme would not provide access to. 
Moreover, a user could trivially wipe the preinstalled operating system and use their own system which does not participate in \scheme. 
We discuss practical means of incentivizing cooperation in Section \ref{ext:participation}.

\subsection{Strawman Proposal}
With the necessary context and assumptions in place, we are now ready to describe an initial 
attempt at a solution which addresses the aforementioned design goals. 
The \scheme framework is divided into four stages which we outline below. 

\smparagraph{Initial Setup}
For simplicity of presentation, assume a set of custodians have among them shares of a secret key $\skj$ such that they can collectively decrypt a ciphertext encrypted under a public key $\pkj$ and sign signatures which verify under a public verification key $\vkj$. 

\smparagraph{Device Setup}
Each manufactured device uses its secure processor to generates a public verification key $\vkd$ along with a secret signing key $\skd$. 
The secret key $\skd$ is stored on the device using \emph{secure sealed storage\footnote{For an example see \url{https://developer.apple.com/documentation/security/certificate_key_and_trust_services/keys/storing_keys_in_the_secure_enclave} for details.}}, managed by the device's secure processor.
The public verification key is published to a public list $\keydb$ maintained by the custodians along with a unique device identifier that can be used to physically track down a device.

\smparagraph{Device Protocol}
Each device selects a randomly chosen subset $\delsel$ of verification keys from $\keydb$ and samples a random nonce $\nonce$. 
To ensure that custodians are required to approve exceptional access requests, the device encrypts the entire set of verification keys it selected using the public key $\pkj$. 
Recall that custodians hold the corresponding shares of the secret key $\skj$. 
The device then outputs the challenge $\chal \define \mathsf{Enc}_{\pkj}(\delsel||\nonce)$ consisting of the encrypted selection along with the nonce. 
For the sake of the strawman argument, we assume that the device is equipped with a read-only mailbox interface accessible via hardware pins~\cite{apple2018t2chip}, which is used exclusively for interfacing with \scheme, and which enables law enforcement to read $\chal$ given physical possession of to the device.

 \smparagraph{Exceptional Access Protocol}
 If law enforcement has physical access to a device, they can read the challenge $\chal$ from the enclave via the secure mailbox interface. 
 Law enforcement then requests the custodians to decrypt $\enc{\delsel||\nonce}$, which, if they agree, allows law enforcement to recover the set of delegate identities selected by the device. 
 Law enforcement then uses the identifying information stored in $\keydb$ to locate the devices selected in $\delsel$. 
 For each device in the set, law enforcement requests that the delegate device to sign the nonce $\nonce$ (output as part of the challenge $\chal$ by the locked device). 
 Call this set of delegate signatures on the nonce $\delselmal$. 
 
 The locked device, when presented with $\delselmal$, locally verifies that the signatures are valid (and match its selection $\delsel$). 
 If the check verifies, the device outputs an access token $\tok$ which allows law enforcement to unlock the device.
\[
	\text{------------------------------}
\]

The high-level strawman solution satisfies the design goals outlined above. Specifically, 
1) authorization is satisfied given that custodians must agree to decrypt the challenge and reveal the identity of the selected delegation; 
2) Non-scalability is achieved by the fact that each device that law enforcement needs unlocked requires a \emph{unique challenge} making the unlock protocol effective only on a \emph{per device} basis; 
3) the physical access requirement is achieved by writing the access challenge to a read-only hardware interface; 
4) transparency is achieved by requiring the involvement of both custodians and randomly chosen set of delegate devices; 
5) social costs are incurred by the need for Law Enforcement to locate a randomly chosen delegation, on a per device basis.

\smparagraph{Problems with Strawman} While the strawman solution described above does achieve the desired properties of Section~\ref{sec:overview}, there are technical questions that must be addressed and fleshed out. Specifically, how is the list $\keydb$ of peer device keys maintained and how does each device obtain it? 
How do we prevent an adversary from inserting malicious delegate keys into the list? 
How do we ensure that devices chosen in the delegation are active (i.e., not lost or broken) and can be located? 
Finally, given the physical access requirement, how can the protocol function well in practice if devices are spread across the world? 
For example, if a device is seized in San Francisco, California, can it securely select a delegation consisting of devices located in, say, the West Coast of the United States? 
We address all these questions in the following sections.

\section{The \scheme Scheme} 
\label{sec:scheme}
In this section, we describe the necessary protocols for realizing \scheme and address the problems described in the strawman scheme. 
To do so, we must introduce some necessary cryptographic and hardware security building blocks we use in our construction. 

\subsection{Cryptographic Building Blocks}
\label{subsec:prelims}
\smparagraph{Secure Processors}
Secure processors such as Apple's T2 Chip~\cite{apple2018t2chip} used in iPhones and MacBooks, and Google's Titan M Chip for high end Android phones~\cite{barrett2018titanchip}, are becoming increasingly prevalent in smartphones and computers alike. 
Such processors enable computations to run securely on the device, isolated from other device components. This enables secure computation even in the presence of adversarial software on the device~\cite{costan2017secure, costan2016sanctum, sau2017survey}. 
What's more, this family of processors are resilient to both hardware and software attacks and engineered to be tamper-resistant~\cite{costan2017secure, costan2016sanctum, sau2017survey}. 
For ease of presentation, we describe our protocol using Apple's T2 terminology, and we make use of Assumption~\ref{assumption:enclave} that the secure processor on the device is \emph{secure}.

\smparagraph{Cryptographic Signatures}
We require a cryptographic signature scheme~\cite{cramer2000signature,hoffstein2001nss} for the purpose of verifying different entities in the system. 
We also require \emph{group} signatures~\cite{boneh2004short}, where a set of entities can generate a signature using shares of a secret key, such that the signature can be verified by a single public key. 

We denote a signature issued by an entity $\mathcal{E}$ in the system by $\sig_{\mathcal{E}}(m)$ where $m$ is the message being signed. 
Anyone can verify the signature $\sig_{\mathcal{E}}(m)$ using the public verification key $\vk_\mathcal{E}$. 

\smparagraph{Merkle Trees}
We use Merkle trees~\cite{merkle1989certified, merkle1987digital} to efficiently validate entries in a list. A Merkle tree consists of three algorithms: $\MerkleDigest$, $\MerklePath$, and $\MerkleVerify$. 
A list of values is hashed using $\MerkleDigest$ which outputs a succinct ``root hash'' of the entire list. A proof that an element is contained in the list can be issued using 
$\MerklePath$ which takes as input an entry from the hashed list and outputs a list of hashes that form a path from the root hash to the entry. 
The proof can be verified efficiently using $\MerkleVerify$ which outputs $1$ if it is given a valid path and $0$ otherwise (i.e., if the element is not in the list). 
Importantly the size of the proof is logarithmic in the size of the list. 

%%%%%%%%%%%%%%%%%%%%%%%%%%%%%%%%%%%%%%%%%%%%%%
\begin{figure}
{
\def\arraystretch{1.25}%  1 is the default, change whatever you need
\begin{tabular}{ |c|p{6.5cm}| }
 \multicolumn{2}{c}{\bf Relevant Notation} \\
 \hline
Symbol & Description\\
 \hline
$\deviceid$          & unique IMEI of a device (see Section~\ref{subsec:imei}). \\
$\keydb$             & set of delegate public keys and encrypted IMEIs. \\
$\numd$              & number of delegate devices in the current epoch. \\
$\numdel$            & number of devices in a delegation. \\
$\numcusto$          & number of custodians in the system. \\
$\vk_\mathcal{E}$    & digital signature verification key of entity $\mathcal{E}$. \\
$\sig_{\mathcal{E}}(\cdot)$ & signature under the signing key of entity $\mathcal{E}$.\\
$\pk_\mathcal{E}$    & public encryption key of entity $\mathcal{E}$. \\
$\enc{v}$            & encryption of a value $v$.\\
$\manu$              & device manufacturer. \\
$\device$            & a delegate device in the system. \\
$\locdevice$         & the locked device which law enforcement wants to access. \\
\hline
\end{tabular}
}	
\end{figure}
%%%%%%%%%%%%%%%%%%%%%%%%%%%%%%%%%%%%%%%%%%%%%%

\subsection{International Mobile Equipment Identity (IMEI)}
\label{subsec:imei}
An IMEI is a unique 15-digit identifier used to identify all mobile phones. 
The IMEI is a \emph{device specific} identifier that is independent of the operating system and cellular network~\cite{scourias1995overview, papadimitratos2006privacy}. 
All devices in a cellular network share their IMEI number with the cellular network provider (e.g., AT\&T) which links it to an International Mobile Subscriber Identity (IMSI)~\cite{scourias1995overview, papadimitratos2006privacy}. 
In this work, we use the fact that network providers know the (IMEI, IMSI) pair for each device in their network and can therefore provide information about the subscriber (i.e., device owner) based on the IMEI. 
For simplicity, we assume that law enforcement already has efficient means to access these records. 
This is a reasonable requirement in practicegiven that the network provider already has the ability to \emph{blacklist} devices based on the IMEI (e.g., if a device is stolen, etc.)~\cite{scourias1995overview, papadimitratos2006privacy}.

\subsection{Epochs}
In order to ensure that most devices in $\keydb$ are operational and that new devices can be added, \scheme is broken into sequential epochs $\epoch \in \{0, 1 \dots\}$.
Each epoch last for some fix amount of time which could be on the order of a week or a month.

During an epoch devices will update themselves to the current list $\keydb_\epoch$ and select delegates from that list.
If a device is issued an unlock request it will use the most recent epoch that it has updated to.
A device does \emph{not} need to be on the current epoch if the device has yet to update. 
Moreover, it is acceptable to have some devices miss epochs due to network failures, power loss, or any number of other reasons.

During epoch $\epoch$, devices will register themselves for the next epoch $\epoch + 1$. At the end of epoch $\epoch$, the custodians will publish a new list $\keydb_{\epoch+1}$ using all devices that have registered during the last epoch. 
This way each device only needs network access at some point during an epoch.

For notational simplicity, we drop the $\epoch$ subscript and denote the number of devices in the current epoch by $\numd$ and the list of delegates by $\keydb$.

\subsection{Custodian Setup}
The set of custodians run a Distributed Key Generation (DKG) protocol~\cite{boneh2001efficient, gennaro2007secure} to obtain a public encryption and verification keys $(\pk_k, \vkj)$ such that each custodian holds a share of the secret key $\skj$.  

Custodians collectively maintain a list of delegates $\keydb$ to which devices can add themselves.
Custodians initialize the first epoch $\epoch = 0$ via the {\bf \CustodianConsensus} protocol.

\subsection{Manufacturer Setup}
\label{subsec:manusetup}
We split manufacturer setup into two protocols: 

\noindent{\bf\SystemSetup} and {\bf \DeviceSetup}. 
The former is used to initialize the system once (i.e., when the manufacturer begins operations), the latter protocol is used to initialize each new manufactured device. \newline

%%%%%%%%%%%%%%%%%%%
% System Setup
%%%%%%%%%%%%%%%%%%%
{\noindent\bf \SystemSetup: }
The manufacturer generates a key pair $(\vkm, \skm)$ for a cryptographic signature scheme. 
The public verification key $\vkm$ is made available to all entities in \scheme. 
We note that, in many cases, this setup is already in place as it is needed for manufacturers to sign software updates~\cite{codesigning}.
Therefore, \scheme can be setup on the manufacturer sideusing existing code-signing infrastructure. 
The manufacturer signs the generated public keys of the custodians, $\pkj$ and $\vkj$, which it outputs to all entities in \scheme. 
\newline

%%%%%%%%%%%%%%%%%%%
% Device Setup
%%%%%%%%%%%%%%%%%%%
{\noindent\bf\DeviceSetup: }
For each manufactured device, the manufacturer interfaces with the device's secure processor via a designated API to request the generation of a new signature key pair. 
The secure processor generates and stores the secret key $\skd$ in \emph{sealed storage}\footnote{Secure encrypted storage managed by the enclave -- $\skd$ never leaves the enclave and remains unknown to the manufacturer.} and outputs the corresponding public verification key $\vkd$ via the secure processor interface. 
The manufacturer signs the device's public key to generate a signature $\sig_{\manu}(\vkd)$ which it sends back to the device's secure processor. 
This signature is subsequently used by the device to prove that it is real device (built by the manufacturer) and not a virtual sybil masquerading as a delegate. 
This process is analogous to existing methods for hardware attestation and widely available at scale~\cite{brickell2004direct, costan2017secure, costan2016sanctum, apple2018t2chip}. 

\subsection{Custodian Protocol}
Over the duration of epoch $\epoch - 1$, custodians accept $\delid{\device}$s from devices. At the end of epoch $\epoch - 1$, custodians collectively perform {\bf CustodianConsensus()} to start a new epoch $\epoch$.\newline

{\noindent \bf\CustodianConsensus:}
At the start of epoch $\epoch$ each custodian has received some number of $\delid{\device}$s from devices. 
Each custodian then validates each $\delid{\device}$ by ensuring it is signed by the manufacture and is not a duplicate.
It then generates a prospective list $\tilde \keydb_\epoch$ of all valid, non-duplicate devices. 
Each custodian then send its prospective list to every other custodian.
Custodians validate every prospective list generated by all other custodians.
Finally, custodians collectively publish the union of all valid devices from each $\tilde \keydb_\epoch$ as $\keydb_\epoch$ and collectively sign it.
Additionally they publish and sign a small metadata header:
\[
  \keydbheader \define \headerepoch
\]
consisting of the root Merkle hash, the number of devices in the epoch, and the current epoch number. 

If at least one custodian does not agree on the collective pick for $\keydb_\epoch$, then everyone aborts the procedure and publishes its chosen list.
This output can be used as a proof of whether a custodian has attempted to publish a non valid list.
Note, that if a malicious custodian chooses to abort the procedure it simply stops the creation of a new epoch. \scheme will still be as secure as the last epoch $\epoch -1$.

\subsection{Device Protocols}
We denote by $\locdevice$ the locked device which law enforcement wants to access.
All delegate devices are denoted by $\device$. 
\label{sec:devproto}
Every device must do the following two things: (1) it must make itself ``available'' as a delegate and (2) it must be able to ``approve'' unlock requests for other devices that have selected it to be part of a delegation. 
We require a protocol by which each device can share some identifying information which can only be revealed by a coalition of custodians. 
Next, we require a protocol for a locked device to securely select a set of delegates without risking man-in-the-middle attacks and without leaking the delegation identity in the process. 
Finally, we require an unlocking protocol by which the locked device reveals a unique access token $\tok$, such as the password, used to unlock it. 

We split these three separate functionalities into the following protocols: {\bf \DelegateRegister,~\SelectDelegation} and {\bf \DeviceUnlock}. \newline

%%%%%%%%%%%%%%%%%%%
% D.DelegateRegister
%%%%%%%%%%%%%%%%%%%
{\noindent\bf $\bm{\device}$.\DelegateRegister$()$: }
The device $\device$ generates a fresh key pair $(\vkd, \skd)$ and encrypts its own $\deviceid$ number using the public key of the custodians to obtain $\enc{\deviceid} \define \mathsf{Enc}_{\pkj}(\deviceid)$. 
The device then generates the tuple
\[
	\delid{\device} \define (\vkd, \enc{\deviceid})
\]
and proceeds to write $\delid{\device}$ to $\keydb$ by sending $\delid{\device}$ to each custodian over a secure channel.
\newline

%%%%%%%%%%%%%%%%%%%
% D.SelectDelegation
%%%%%%%%%%%%%%%%%%%
{\noindent\bf $\bm{\locdevice}$.\SelectDelegation$(\keydbheader )$: }
The device downloads the signed header $\keydbheader$ of the most recent $\keydb_\epoch$. Recall 
  \[
  \keydbheader  \define \header
  \]
The device verifies the header signature using $\vkj$. It then selects indices $\ell_1,\ell_2,\dots,\ell_\numdel$ uniformly at random such that $1\le \ell_i \le \numd$, where $\numd$ is the size of $\keydb$. 
Let $\delsel = \{\ell_1, \dots, \ell_\numdel\}$ be a set of indices which the device stores in secure memory. This \emph{delegate selection} does not leave the device.
\newline

%%%%%%%%%%%%%%%%%%%
% D.RevealDelegation
%%%%%%%%%%%%%%%%%%%
{\noindent\bf $\bm{\locdevice}$.\RevealChallenge$()$: } 
The locked device $\locdevice$ generates a random nonce $\nonce$ and encrypts its selection of delegates along with the nonce under the custodian's encryption key $\enc{\delsel||\nonce} \define \mathsf{Enc}_{\pkj}(\delsel||\nonce)$, and outputs the challenge
\[
\chal \define \enc{\delsel||\nonce}
\]
via the enclaves's secure mailbox interface. 
This protocol is only used if a device is attempted to be unlocked via \scheme. 
Once a device outputs an unlock challenge, it should not update to new epochs for some fixed period to ensure reasonable time to complete the request.
\newline

%%%%%%%%%%%%%%%%%%%
% D.SignRequest
%%%%%%%%%%%%%%%%%%%
{\noindent\bf $\bm{\device}$.\DelegateSignRequest$( \nonce )$: }
The delegate device $\device$ signs the input nonce and sends back $\sig_{\device}(\nonce)$ through the secure hardware mailbox interface. 
This process allows law enforcement to obtain the necessary signatures from delegate devices without compromising the security and privacy of delegates.
\newline

%%%%%%%%%%%%%%%%%%%
% D.RevealUnlockToken
%%%%%%%%%%%%%%%%%%%
{\noindent\bf $\bm{\locdevice}$.\DeviceUnlock$(\delselmal, \MerkleProof)$: } 
When given the set $\delselmal = \{j||\delid{\device_j}||\sig_{\device_j}(\nonce)~|~ j \in \delsel\}$ 
consisting of verification keys (recall that $\delid{\device} \define (\vk_\device, \enc{\deviceid})$) and signatures obtained from a subset of delegate devices in $\keydb$, indexed by the set $\delsel$, and $\MerkleProof = \{\pi_1, \dots, \pi_{|\delselmal|}\}$ consisting of a set of Merkle proofs for each delegate verification key, the device verifies the following:

\begin{enumerate}[label=(\arabic*)]
	\item $|\delselmal| \ge \threshdel$ where $\threshdel$ is the threshold of 
	signatures required for unlocking;
	
	\item $\MerkleVerify(\pi_j, \delid{\device_j}||j) = 1 $ for each element in $\delselmal$;
	
	\item For each signature $\sig_{\device_j}(\nonce) \in \delselmal$, the corresponding verification key $\vk_{\device_j}$ verifies the signature. 
\end{enumerate}

\noindent The verifications ensure that a sufficient coalition of devices from the delegate selection signed the unlock request and, crucially, that each signature obtained was from delegate device indexed by the set $\delsel$ in $\keydb$. 
This latter step ensures that the delegation selection was not ``swapped out'' for a selection $\delsel^* \neq \delsel$.

If the three verification steps pass, $\locdevice$ reveals the unlock token\footnote{The token $\tok$ can be as simple as the device's password.} $\tok$ by writing it to the secure mailbox interface. 
Otherwise, the unlock process fails and $\locdevice$ outputs $\bot$. 

\subsection{Delegate Protocol}
To participate in the system as a delegate, a device $\device$ must add itself to the list of participating devices. 
For each new time epoch, $\device$ runs {\bf\DelegateRegister} to add itself to the list $\keydb$ and makes itself available for selection. 

Devices can be called upon as delegates in order to verify the unlock of other devices through the {\bf\DelegateSignRequest} protocol. 
We stress that the process by which an owner of a delegate devices signs a nonce should be \emph{non-invasive}.
Specifically, there is no need for a device owner in a delegation to unlock their own device (which is acting as a delegate) in order to sign a request. 
The signature can be obtained by interfacing directly with the device's enclave through a physical pin connection or wirelessly using built in RFID.
However, physical access by law enforcement is needed both to incur a social cost and to reduce the risk of side channel attacks but \emph{participation should not incur an invasion of privacy for the delegate device owner.}

%%%%%%%%%%%%%%%%%%%
% full protocol
%%%%%%%%%%%%%%%%%%%
\subsection{Warranted Access Protocol} \label{scheme:wap}
We will now describe how law enforcement proceeds to unlock a device through \scheme. 
We assume that law enforcement has \emph{physical access} to $\locdevice$ so that they can interface with the secure processor over a mailbox interface~\cite{costan2016sanctum}. 
Law enforcement proceeds as follows:\newline 

\noindent\textbf{Step 1:} Law enforcement sends a request to the secure hardware enclave via $\RevealChallenge$ and obtains
the challenge $\chal \define \enc{\delsel||\nonce}$. 
    
\noindent\textbf{Step 2:} Law enforcement now presents $\chal$ to the custodians. 
The custodians decrypt $\enc{\delsel||\nonce}$ and reveal the IMEI numbers of the selected delegation by decrypting the entries in $\keydb$ indexed by $\delsel$. 
Recall that the $i$th entry in $\keydb$ is of the form $\keydb[i] = (\vk_i||\enc{\deviceid_{\device_i}}||i)$.
The custodians output the set
\[
	\ids \define \{ \deviceid_{\device_j} ~|~ j \in \delsel \}
\]
\noindent\textbf{Step 3:} Law enforcement now has the IMEI numbers, which uniquely identify each device in the delegation. 
Law enforcement is now responsible for physically locating devices in $\ids$ to obtain a signature. To prove that the selected delegate verification keys in $\ids$ are correct (i.e., are indeed the keys in $\keydb$ indexed by $\delsel$), law enforcement generates a set of Merkle proofs,
\[
	\MerkleProof = \{ {\tt MerklePath}(\keydb[j]) ~|~ j \in \delsel \}
\]
where $\keydb[j]$ is the element at the $j$th index of $\keydb$.  

\noindent\textbf{Step 4:} Law enforcement needs to locate at least $\threshdel$ devices in the set. 
For each delegate device that law enforcement locates, they interface with the secure hardware mailbox interface to send $\nonce$ to the device. 
The delegate device responds  according to {\bf \DelegateSignRequest}.
Define the set of all $\threshdel$ delegate verification keys and signatures of the nonce
\[
	\delselmal = \{i||\delid{\device_i}||\sig_{{\device_i}}(\nonce) ~|~ \text{ for } i \in \delsel \}.
\]
\noindent\textbf{Step 5:} Law enforcement then sends $(\delselmal, \MerkleProof)$ to $\locdevice$ in response to the challenge $\chal$. 
$\locdevice$ verifies the correctness of the response according to $\locdevice.\DeviceUnlock(\delselmal, \MerkleProof)$.

% !TEX root = main.tex
\section{Security \& Privacy}
\label{sec:security}
In this section we analyze the security properties of \scheme in relation to our threat model. 
We then address several privacy concerns which are important to consider in a practical system. 

\subsection{Security Analysis}\label{sec:securityAnalysis}
We present the main theorems we use to prove security in this section but 
defer the proofs to Appendix \ref{apdx:proofs}.

% definition 1
\begin{definition} An \textbf{active device} for an epoch $\epoch$ is a device which successfully transmits its public key to at least one honest custodian. \scheme assumes that a sufficient  number of devices will be active for each epoch even if many devices are not. See Section \ref{sec:practicalSecurity}
\end{definition}

% definition 2
\begin{definition}
	An \textbf{honest} custodian or device is one which correctly follows the protocol defined in Section \ref{sec:scheme} and which is not corrupted by an adversary.
\end{definition}

% theorem 1
\begin{restatable}{theorem}{thmsetup}
\label{thm:setup}
Assume a secure public key infrastructure and that at least one custodian is honest.
The procedure $\CustodianConsensus$ for epoch $\epoch$ will either abort or result in a published list $\keydb_\epoch$ and root hash $\keydbheader_\epoch$ which contains entries from each honest active device and only entries from devices signed by the manufacturer $\manu$.
\end{restatable}

% lemma 1
\begin{restatable}{lemma}{lemmacustpriv}
\label{lem:custpriv}
Let $\device$ be an honest device with delegation selection $\delsel_\epoch \subseteq \keydb_\epoch$ for epoch $\epoch$.
Then, for all PPT adversaries $\adv$ if there is at least one honest custodian, $\adv$ cannot recover $\delsel$ with probability non-negligibly better than random guessing.
\end{restatable}

% lemma 2
\begin{restatable}{lemma}{lemmaneedkeys}
\label{lem:needkeys}
Suppose Theorem \ref{thm:setup} holds for $\keydb_\epoch$ in epoch $\epoch$.
Let $\device$ be a device with delegate selection $\delsel_\epoch \subseteq \keydb_\epoch$ and which has downloaded the header $\keydbheader_\epoch$ of the epoch.
Then for all PPT adversaries $\adv$ without access to $\{ \sk_i ~|~ i \in \delsel\}$, the probability that $\adv$ can unlock $\device$ using \scheme is negligible in the security parameter of the signature scheme.
\end{restatable}

% lemma 3
\begin{restatable}{lemma}{lemmafrac}
\label{lem:frac}
Suppose Theorem \ref{thm:setup} holds for $\keydb_\epoch$ in epoch $\epoch$. Let $\numd = |\keydb_\epoch|$, $\delsel_\epoch$ be the delegate selection of a device and $\adv$ be a PPT adversary. 
Suppose $\adv$ has corrupted at most $\pcorrupt \in (0, \numd)$ devices.

Then, for any probability $\secprob \in (0, 1)$, there exists a delegation size $\numdel$ and threshold $\threshdel$ such that the probability that $\adv$ has corrupted every device in $\delsel_\epoch$ is less than $p$.
\end{restatable}

% threom 2
\begin{restatable}{theorem}{thmmain}
\label{thm:main}
Suppose Theorem \ref{thm:setup} holds for $\keydb_\epoch$ in epoch $\epoch$.
Let $\adv$ be a PPT adversary and $\numd = |\keydb_\epoch|$.

Suppose $\adv$ has corrupted at most $\pcorrupt \in (0, \numd)$ devices. 
Let $\locdevice$ be a device for which $\adv$ is then given a decrypted {\upshape \chal} (for instance from the custodians) after having selected devices to corrupt.
For any $\secprob \in (0,1)$ there exists a delegation size $\numdel$ and threshold $\threshdel$ such that for every such $\locdevice$ there is an independent probability bounded by $p$ that $\adv$ can unlock $\locdevice$ using \scheme without cooperation of an honest delegate.
\end{restatable}

\smparagraph{Privacy Analysis}
It is also important to show that reasonable privacy for devices is maintained by \scheme. 
Recall that the $i$th entry of the list $\keydb$, denoted $\delid{i}$, is

\[
	\delid{i} \define (\vk_i, \enc{\deviceid_{\device_i}})
\]
and corresponds to some $i$th delegate device $\device_i$. Each $\vk_i$ is freshly generated for each new epoch and the publicly identifiable information (the $\deviceid_{\device_i}$ value) is encrypted under the custodians threshold scheme. 
 It is reasonable, however, to assume that a sophisticated adversary could link device identities to $\keydb$ entires through network traffic analysis. 
 This, however, does  not compromise the security of the scheme for honest devices by Lemma \ref{lem:custpriv} and only leaks that a particular device is participating in \scheme.

%%%%%%%%%%%%%%%%%%%%%%%%%%%%%%%%%%%%%%%
\begin{figure}
\centering
  \includegraphics[width=\linewidth]{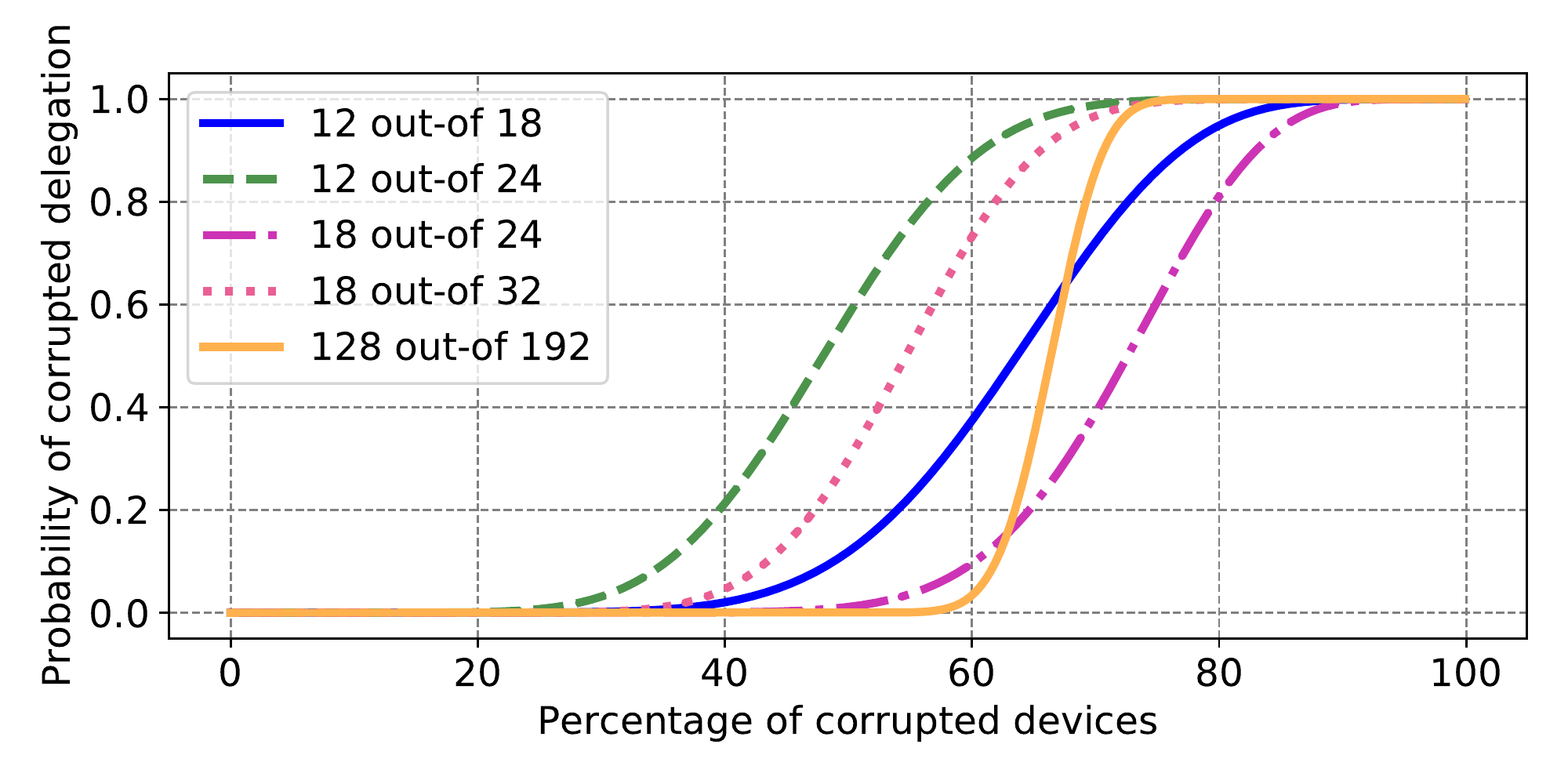}
  \caption{The probability that an adversary will have access to all needed delegate keys given 
  the percentage of all devices it controls in the current epoch. Each curve represents a different threshold of 
  key shares needed for the exceptional access scheme. }
  \label{fig:corrupt}
\end{figure}
%%%%%%%%%%%%%%%%%%%%%%%%%%%%%%%%%%%%%%%

\subsection{Practical Security and Concerns} \label{sec:practicalSecurity}

While our proofs of security provide necessary guarantees, there are practical concerns that should be addressed.

\smparagraph{Partial Epoch Participation} 
Devices may miss epochs due to having limited network access, no power, or for any number of other reasons. However, as long as most devices participate in an epoch $\epoch$, the size of $\keydb$ will be sufficient for security. 
The list $\keydb$ for each epoch can be maintained by the custodians and each device can maintain its old verification keys. Devices which have not updated to new epochs can be unlocked using \scheme and the header of the last valid epoch.

\smparagraph{Lost or Uncooperative Delegate devices} For \scheme to be effective, 
it needs to work even if devices on the selected delegation are lost, broken, out of the country or uncooperative. 
In \scheme, we fix a threshold $\threshdel$ of $\numdel$ signatures needed for unlocking. 
Law enforcement needs only to access some $\threshdel$ devices to proceed with an unlock request. 
Therefore, even if $\numdel-\threshdel$ devices are inaccessible, the unlock protocol is able to complete successfully. Additionally, \scheme refreshes the delegation of devices for each epoch (e.g., on a weekly basis), therefore new delegations are comprised only of active devices.

\smparagraph{Practical Levels of Corrupted Devices}
We show in Theorem~\ref{thm:main} that \scheme can be constructed so that probability of an adversary possessing the needed signatures for a single device is arbitrarily small. 
However, we need to show that this is secure even given practical values for the parameters $\threshdel$ and $\numdel$. 

Assuming a device choses its delegation of devices uniformly at random, the probability that an adversary will possess all keys in a delegation is the sum of the probability that it will possess exactly $i \geq \threshdel$ shares multiplied by the number of combinations it can achieve this in. Formally, 

\begin{equation}
\label{eq:practicalcorrupt}
	 p = \sum_{i=\threshdel}^{\numdel} \binom{\numdel}{i} \left( \frac{\pcorrupt}{\numd} 
	 \right)^i \left(1 - \frac{\pcorrupt}{\numd}  \right)^{\numdel-i} 
\end{equation}
We plot Equation~\ref{eq:practicalcorrupt} in Figure~\ref{fig:corrupt} for five values of $\numdel$ and $\threshdel$. 
This is the probability that an adversary can unlock an arbitrary honest device given it controls some percent of all devices. 
We stress that this success probability is independent for each device (e.g. being able to unlock a specific device does not help an adversary unlock other devices).
For reasonable parameters, such as $\numdel=18$ and $t=12$, we obtain $\secprob < 0.01$ even when the adversary controls $35\%$ of all devices. 
For more conservative values, such as $15\%$ of all devices being corrupted, we obtain a probability of $\secprob < 10^{-6}$ which is roughly the same as guessing a $6$ digit pin code.

\smparagraph{Replay Attacks}
When the device generates its unlock $\chal$ through {\bf \DeviceUnlock} it creates a devices specific random nonce $\nonce$.
This nonce is signed by each delegate to complete the unlock challenge. 
In practice, delegate signatures cannot be reused by law enforcement to unlock other devices with different nonces.
In our security analysis in Section \ref{sec:securityAnalysis}; however, we assume that once a device has been corrupted the adversary can use it to sign new requests.

\subsection{Auditing}
An advantage of \scheme is that both the manufacturer and custodians can be audited by one another and by the public. This helps ensure that \scheme is being correctly carried out in practice by quickly revealing anomalous behavior. 

\smparagraph{Auditing Manufacturers}
Under our threat model we assume that a device's manufacture needs to be honest-but-curious, otherwise the device itself could be built in a compromised manner via a supply chain attack.
It's possible, however, that the manufacture could be coerced into malicious participation with an adversary after building a device.

A malicious manufacturer could perform a sybil attack by flooding the system with sybil devices under an adversary's control. Since manufacturers validate the device keys they could take over an arbitrary percentage of all devices.
However, the number of devices in \scheme is public knowledge and collectively the custodians have the power to approve the list $\keydb$. 
Any public auditor observing the resulting $\keydb$ for each epoch would be able to discover a significant spike in new devices. 
Since an adversary must control a large number of all devices, we believe even with a malicious manufacturer it would be difficult for such an attack to succeed without notice.

\smparagraph{Auditing Custodians}
Likewise, devices and the public can verify that the custodians are correctly carrying out \scheme. 
The Merkle hash of $\keydb$ guarantees that the correct public keys are used during unlocking. 
Since $\keydb$ is public, however, any device can additionally verify that its key has been added to $\keydb$.

\subsection{Privacy and Transparency}
There are two crucial sides to consider when analyzing the privacy properties of \scheme: privacy of delegate devices (and the people that own the devices) and the transparency of the exceptional access use. 

\smparagraph{Delegation Privacy}
The privacy of delegates is maintained by only having public keys stored in $\keydb$ while all other uniquely identifying information (such as the IMEI number) remains encrypted and only accessible via consensus of the custodians. A delegate obtains no information about the identity of the device being unlocked only that they were selected as one of its delegates. 

However, \scheme does require law enforcement to contact delegates in order unlock a device. 
We argue that this very real inconvenience is a feature and not a bug. 
Using \scheme requires law enforcement to pay this social cost. 
Every use of \scheme fundamentally must leak the use of exceptional access to ordinary citizens. 
If use of the exceptional access scheme were to be frequent, then it is perhaps good that citizens would be aware and discomforted by this abuse of power. 

Moreover, being called to serve on a traditional jury is already an accepted social responsibility in many democracies. One can view participation in \scheme as an extension of jury duty in the digital sphere. 
Additionally, delegate devices can be requested to provide a digital signature without unlocking or exposing personal information to law enforcement, as explained in Section~\ref{sec:design}.

\smparagraph{Law Enforcement Privacy}
Public information pertaining to on-going cases could severely impede investigation for law enforcement~\cite{frankle2018practical}. \scheme reveals nothing more than the fact that \emph{some device is being unlocked} meaning that law enforcement secrecy is maintained.

\section{Extensions}

\subsection{International Jurisdiction} \label{sec:jurisdiction}
\scheme could present difficult international jurisdiction issues if it were to be implemented in multiple countries.
Suppose a device sold in the US was confiscated in Australia.
If Australian authorities request device access they would need physical access to a number of American delegates.
This could result in complicated legal questions.

A simple solution is to have devices update themselves each epoch to the local $\keydb_{local}$ based on geolocation.
This does not completely solve the problem, but it could mitigate it.
Additionally, it reduces the number of devices that are inaccessible due to international travel.

A more complete solution would be to allow any jurisdiction to unlock any device it has physical access to.
There is a danger this would make it easier for international law enforcement to unlock devices, but this could also be desirable.
Custodians could maintain a seperate $\keydb$ for each jurisdiction.
Additionally, they publish a single Merkle root hash of the union of all $\keydb$'s.
Instead of selecting random indices from a single $\keydb$, devices select a random seed.
When law enforcement query {\bf \DeviceUnlock} they present the root hash of the subtree defined by the $\keydb$ of their jurisdiction.
The indices of the device's delegation are defined by the seed and a predefined fixed pseudo-random  generator.
This way indices can be chosen randomly from any $\keydb$ and each jurisdiction has a valid delegation. 

\subsection{Incentivizing \scheme Participation} \label{ext:participation}
As mentioned in Assumption \ref{assumption:honesty}, it is impossible to force a device to cooperate with any exceptional access scheme. 
For instance, a user could always use a second layer of strong encryption. 
In practice, however, the secure processor is the only means to pair strong physical security with ease of use~\cite{apple2018t2chip, costan2016sanctum, vanbulck2018foreshadow, van2018foreshadow}. 
Directly encrypting the device requires the user to memorize or store a full encryption key which is often impractical.

Modern devices using a secure hardware enclave often use a mechanism such as a pin code or biometric to decrypt the device while maintaining state-of-the-art security \cite{apple2018t2chip}. 
\scheme can therefore tie the use of the secure processor with participation. Every part of \scheme that happens on a device happens in its secure enclave. 
As long as the enclave verifies that it continues to receive valid headers for lists $\keydb_\epoch$, for each new epoch $\epoch$, it can be unlocked using \scheme. 
Likewise, the enclave can verify that its writes are received by the custodians and that it is successfully contributing its delegate verification keys. 
If the enclave fails to update to a new epoch, it can continue to provide \scheme exceptional access on the last valid epoch.
	
If too much time passes since the last updated epoch and the device's password has been correctly used during that time, then the secure enclave can provide a failsafe where any valid unlock request signed by the custodians will unlock the device. 
This timeout period can be set to be a reasonably long period of time (such as six months) to incentivize devices to participate.

% !TEX root = main.tex

\section{Limitations} 
\label{sec:limitations}
We believe that \scheme is a significant advancement compared to previously proposed exceptional access schemes, especially when it comes to reducing trust in any single authority and mitigating the risk of abuse. 
However, we are under no illusion that \scheme is capable or solving all security, practical or ethical problems resulting from deploying exceptional access. 
We discuss a few important, but by no means exhaustive, limitations to \scheme and other similar systems.

\smparagraph{Necessary Political Infrastructure}
In order for \scheme to operate, there needs to be reasonable trust in public institutions selected as custodians and public accountability of the government.
\scheme rests on the assumption that it would be difficult for abusive law enforcement to secretively coerce all custodians or to routinely silence citizens selected as delegates. 
This is certainly not true in all countries especially countries with authoritarian or unstable governments. In these situations implementing \scheme is likely unsuitable.

\smparagraph{Jurisdiction of \scheme}
Today, most devices are designed and assembled in supply chains that cross international borders~\cite{tradewar}. 
Additionally, a handful of large corporations under the jurisdiction of a few powerful countries build and sell devices across the entire world~\cite{bigtech}. 
Consider a device which participates in \scheme, but the custodians and most other peer devices resided in another country. 
Attempting to unlock a device using \scheme could result in complicated legal questions of international sovereignty. 
We discuss potential solutions in Section \ref{sec:jurisdiction}.

\smparagraph{Normalization of Exceptional Access}
If a country were to adopt an ethically designed exceptional access scheme such as \scheme, it could provide cover for less judicious countries to employ insidious backdoors with claims of false equivalence. 
Additionally, if citizens objected to participating as delegates in \scheme, dishonest actors might call to remove the delegate requirement (thereby undermining the entire system), rather than be pressured into reducing the use of exceptional access. 
We discuss the broader ethics of researching exceptional access in Section \ref{sec:intro}.

\section{Conclusion}
\scheme is the first exceptional access scheme which does not place approval for unlocks in a single trusted authority, but rather in the hands of devices participating in the scheme. 
We show that \scheme provides strong security and is resistant to abuse for mass surveillance by tying unlocks to a social cost. 
We present a feasible outline of how \scheme could be implemented in practice and rigorously analyze the security of the protocol. We hope that the ideas presented in this paper will spur discussion on what exceptional access schemes are possible and deemed appropriate by society.
 
\section*{Ackownledgments}
We thank Kyle Hogan, Seny Kamara, Ilia Lebedev, Ben Murphy, Sarah Scheffler, and Mayank Varia for helpful discussions on this topic and providing us with many valuable suggestions on early drafts of this paper. 

\newpage
% dont color links in biblio
\hypersetup{urlcolor=black}
\bibliographystyle{plain}
\bibliography{biblio}

\begin{thebibliography}{10}

\bibitem{fiveeyes}
Five Country~Ministerial 2018.
\newblock Statement of principles on access to evidence and encryption.
\newblock
  \url{web.archive.org/web/20180925154820/https://www.homeaffairs.gov.au/about/national-security/five-country-ministerial-2018/access-evidence-encryption},
  2018.

\bibitem{abelson2015keys}
Harold Abelson, Ross Anderson, Steven~M Bellovin, Josh Benaloh, Matt Blaze,
  Whitfield Diffie, John Gilmore, Matthew Green, Susan Landau, Peter~G Neumann,
  et~al.
\newblock Keys under doormats: mandating insecurity by requiring government
  access to all data and communications.
\newblock {\em Journal of Cybersecurity}, 1(1):69--79, 2015.

\bibitem{apple2018t2chip}
Inc. Apple.
\newblock Apple {T2} security chip.
\newblock \url{www.apple.com/mac/docs/Apple_T2_Security_Chip_Overview.pdf},
  2018.

\bibitem{barrett2018titanchip}
Brian Barrett.
\newblock The tiny chip that powers up {Pixel 3} security.
\newblock \url{www.wired.com/author/brian-barrett/}, 2018.

\bibitem{bates2015accountable}
Adam Bates, Kevin~RB Butler, Micah Sherr, Clay Shields, Patrick Traynor, and
  Dan Wallach.
\newblock Accountable wiretapping--or--i know they can hear you now.
\newblock {\em Journal of Computer Security}, 23(2):167--195, 2015.

\bibitem{bellare1996encapsulated}
Mihir Bellare and Shafi Goldwasser.
\newblock Encapsulated key escrow, 1996.

\bibitem{bellare1997verifiable}
Mihir Bellare and Shafi Goldwasser.
\newblock Verifiable partial key escrow.
\newblock In {\em Proceedings of the 4th ACM Conference on Computer and
  Communications Security}, pages 78--91. ACM, 1997.

\bibitem{benner_2020}
Katie Benner.
\newblock Barr asks apple to unlock pensacola killer’s phones, setting up
  clash.
\newblock {\em The New York Times}, Jan 2020.

\bibitem{boneh2018verifiable}
Dan Boneh, Joseph Bonneau, Benedikt B{\"u}nz, and Ben Fisch.
\newblock Verifiable delay functions.
\newblock In {\em Annual International Cryptology Conference}, pages 757--788.
  Springer, 2018.

\bibitem{boneh2004short}
Dan Boneh, Xavier Boyen, and Hovav Shacham.
\newblock Short group signatures.
\newblock In {\em Annual international cryptology conference}, pages 41--55.
  Springer, 2004.

\bibitem{boneh2018survey}
Dan Boneh, Benedikt B{\"u}nz, and Ben Fisch.
\newblock A survey of two verifiable delay functions.
\newblock Technical report, Cryptology ePrint Archive, Report 2018/712, 2018.
  https://eprint. iacr. org~…, 2018.

\bibitem{boneh2001efficient}
Dan Boneh and Matthew Franklin.
\newblock Efficient generation of shared {RSA} keys.
\newblock {\em Journal of the ACM (JACM)}, 48(4):702--722, 2001.

\bibitem{brickell2004direct}
Ernie Brickell, Jan Camenisch, and Liqun Chen.
\newblock Direct anonymous attestation.
\newblock In {\em Proceedings of the 11th ACM conference on Computer and
  communications security}, pages 132--145. ACM, 2004.

\bibitem{buchanan2016cybersecurity}
Ben Buchanan.
\newblock {\em The cybersecurity dilemma: Hacking, trust, and fear between
  nations}.
\newblock Oxford University Press, 2016.

\bibitem{tradewar}
Ben Casselman.
\newblock Trade war starts changing manufacturers in hard-to-reverse ways,
  2019.

\bibitem{castro2019}
Erwin Castro.
\newblock {\em international business times}, Sep 2019.

\bibitem{coker2011principles}
George Coker, Joshua Guttman, Peter Loscocco, Amy Herzog, Jonathan Millen,
  Brian O’Hanlon, John Ramsdell, Ariel Segall, Justin Sheehy, and Brian
  Sniffen.
\newblock Principles of remote attestation.
\newblock {\em International Journal of Information Security}, 10(2):63--81,
  2011.

\bibitem{costan2016sanctum}
Victor Costan, Ilia Lebedev, and Srinivas Devadas.
\newblock {Sanctum:} minimal hardware extensions for strong software isolation.
\newblock In {\em 25th $\{$USENIX$\}$ Security Symposium ($\{$USENIX$\}$
  Security 16)}, pages 857--874, 2016.

\bibitem{costan2017secure}
Victor Costan, Ilia Lebedev, Srinivas Devadas, et~al.
\newblock Secure processors part i: Background, taxonomy for secure enclaves
  and intel sgx architecture.
\newblock {\em Foundations and Trends{\textregistered} in Electronic Design
  Automation}, 11(1-2):1--248, 2017.

\bibitem{cramer2000signature}
Ronald Cramer and Victor Shoup.
\newblock Signature schemes based on the strong {RSA} assumption.
\newblock {\em ACM Transactions on Information and System Security (TISSEC)},
  3(3):161--185, 2000.

\bibitem{edgar2017beyond}
Timothy~H Edgar.
\newblock {\em Beyond Snowden: Privacy, Mass Surveillance, and the Struggle to
  Reform the {NSA}}.
\newblock Brookings Institution Press, 2017.

\bibitem{nsa2013primary}
Electronic~Frontier Foundation.
\newblock {NSA} primary sources.
\newblock \url{{www.eff.org/nsa-spying/nsadocs}}, urldate = {2013-6-5}, 2013.

\bibitem{frankle2018practical}
Jonathan Frankle, Sunoo Park, Daniel Shaar, Shafi Goldwasser, and Daniel
  Weitzner.
\newblock Practical accountability of secret processes.
\newblock In {\em 27th $\{$USENIX$\}$ Security Symposium ($\{$USENIX$\}$
  Security 18)}, pages 657--674, 2018.

\bibitem{gennaro2007secure}
Rosario Gennaro, Stanislaw Jarecki, Hugo Krawczyk, and Tal Rabin.
\newblock Secure distributed key generation for discrete-log based
  cryptosystems.
\newblock {\em Journal of Cryptology}, 20(1):51--83, 2007.

\bibitem{goldwasser1998time}
Shafi Goldwasser and Mihir Bellare.
\newblock Time delayed key escrow, June~16 1998.
\newblock US Patent 5,768,388.

\bibitem{goldwasser2017public}
Shafi Goldwasser and Sunoo Park.
\newblock Public accountability vs. secret laws: Can they coexist?: A
  cryptographic proposal.
\newblock In {\em Proceedings of the 2017 on Workshop on Privacy in the
  Electronic Society}, pages 99--110. ACM, 2017.

\bibitem{wiredglobalpriv}
Lily Hay~Newman.
\newblock Australia's encryption-busting law could impact global privacy.
\newblock \url{www.wired.com/story/australia-encryption-law-global-impact/},
  2018.

\bibitem{hoffstein2001nss}
Jeffrey Hoffstein, Jill Pipher, and Joseph~H Silverman.
\newblock {NSS:} an {NTRU} lattice-based signature scheme.
\newblock In {\em International Conference on the Theory and Applications of
  Cryptographic Techniques}, pages 211--228. Springer, 2001.

\bibitem{wjs2016applevsfbi}
The Wall~Street Journal.
\newblock The {FBI} vs. {Apple}.
\newblock \url{www.wsj.com/articles/the-fbi-vs-apple-1455840721}, 2016.

\bibitem{levycracking}
Jacob Kastrenakes.
\newblock Cracking the crypto war.
\newblock \url{www.wired.com/story/crypto-war-clear-encryption/}, 2018.

\bibitem{kroll2014secure}
Joshua~A Kroll, Edward~W Felten, and Dan Boneh.
\newblock Secure protocols for accountable warrant execution.
\newblock {\em See \url{www. cs. princeton. edu/felten/warrant-paper. pdf}},
  2014.

\bibitem{mandt2016demystifying}
Tarjei Mandt, Mathew Solnik, and David Wang.
\newblock Demystifying the secure enclave processor.
\newblock {\em Black Hat Las Vegas}, 2016.

\bibitem{bigtech}
David McLaughlin.
\newblock Did big tech get too big? more of the world is asking, 2019.

\bibitem{merkle1987digital}
Ralph~C Merkle.
\newblock A digital signature based on a conventional encryption function.
\newblock In {\em Conference on the theory and application of cryptographic
  techniques}, pages 369--378. Springer, 1987.

\bibitem{merkle1989certified}
Ralph~C Merkle.
\newblock A certified digital signature.
\newblock In {\em Conference on the Theory and Application of Cryptology},
  pages 218--238. Springer, 1989.

\bibitem{codesigning}
Microsoft.
\newblock Introduction to code signing.
\newblock
  \url{docs.microsoft.com/en-us/previous-versions/windows/internet-explorer/ie-developer/platform-apis/ms537361(v=vs.85)},
  2017.

\bibitem{nakashima2013chinese}
Ellen Nakashima.
\newblock {Chinese hackers who breached {Google} gained access to sensitive
  data, {U.S.} officials say}, 2013.

\bibitem{newman2018australia}
Lily~Hay Newman.
\newblock Australia's encryption-busting law could impact global privacy, 2018.

\bibitem{bbc2018australia}
BBC News.
\newblock Australia data encryption laws explained.
\newblock \url{www.bbc.com/news/world-australia-46463029}, 2018.

\bibitem{ozzieclear}
Ray Ozzie.
\newblock {CLEAR}.
\newblock \url{github.com/rayozzie/clear/blob/master/clear-rozzie.pdf}, 2017.

\bibitem{papadimitratos2006privacy}
Panagiotis Papadimitratos, Antonio Kung, Jean-Pierre Hubaux, and Frank Kargl.
\newblock Privacy and identity management for vehicular communication systems:
  a position paper.
\newblock In {\em Workshop on standards for privacy in user-centric identity
  management}, number CONF, 2006.

\bibitem{sau2017survey}
Suman Sau, Jawad Haj-Yahya, Ming~Ming Wong, Kwok~Yan Lam, and Anupam
  Chattopadhyay.
\newblock Survey of secure processors.
\newblock In {\em Embedded Computer Systems: Architectures, Modeling, and
  Simulation (SAMOS), 2017 International Conference on}, pages 253--260. IEEE,
  2017.

\bibitem{savage2018lawful}
Stefan Savage.
\newblock Lawful device access without mass surveillance risk: A technical
  design discussion.
\newblock In {\em Proceedings of the 2018 ACM SIGSAC Conference on Computer and
  Communications Security}, pages 1761--1774. ACM, 2018.

\bibitem{scourias1995overview}
John Scourias.
\newblock Overview of the global system for mobile communications.
\newblock {\em University of Waterloo}, 4, 1995.

\bibitem{spreitzer2018systematic}
Raphael Spreitzer, Veelasha Moonsamy, Thomas Korak, and Stefan Mangard.
\newblock Systematic classification of side-channel attacks: a case study for
  mobile devices.
\newblock 2018.

\bibitem{van2018foreshadow}
Jo~Van~Bulck, Marina Minkin, Ofir Weisse, Daniel Genkin, Baris Kasikci, Frank
  Piessens, Mark Silberstein, Thomas~F Wenisch, Yuval Yarom, and Raoul Strackx.
\newblock Foreshadow: Extracting the keys to the {Intel} {SGX} kingdom with
  transient out-of-order execution.
\newblock In {\em 27th $\{$USENIX$\}$ Security Symposium ($\{$USENIX$\}$
  Security 18)}, pages 991--1008, 2018.

\bibitem{vanbulck2018foreshadow}
Jo~Van~Bulck, Marina Minkin, Ofir Weisse, Daniel Genkin, Baris Kasikci, Frank
  Piessens, Mark Silberstein, Thomas~F. Wenisch, Yuval Yarom, and Raoul
  Strackx.
\newblock Foreshadow: Extracting the keys to the {Intel SGX} kingdom with
  transient out-of-order execution.
\newblock In {\em Proceedings of the 27th {USENIX} Security Symposium}.
  {USENIX} Association, August 2018.
\newblock See also technical report
  Foreshadow-NG~\cite{weisse2018foreshadowNG}.

\bibitem{weisse2018foreshadowNG}
Ofir Weisse, Jo~Van~Bulck, Marina Minkin, Daniel Genkin, Baris Kasikci, Frank
  Piessens, Mark Silberstein, Raoul Strackx, Thomas~F Wenisch, and Yuval Yarom.
\newblock {Foreshadow-NG:} breaking the virtual memory abstraction with
  transient out-of-order execution.
\newblock 2018.

\bibitem{weitzner2018encryption}
Daniel~J. Weitzner.
\newblock Perspectives on encryption and surveillance.
\newblock \url{www.lawfareblog.com/perspectives-encryption-and-surveillance},
  2018.

\bibitem{wright2018crypto}
Charles Wright and Mayank Varia.
\newblock Crypto crumple zones: Enabling limited access without mass
  surveillance.
\newblock In {\em 2018 IEEE European Symposium on Security and Privacy
  (EuroS\&P)}, pages 288--306. IEEE, 2018.

\end{thebibliography}
\appendix
% !TEX root = main.tex

\section{Proofs}
\label{apdx:proofs}
 
\thmsetup*
  
\begin{proof}

Let $\device$ be an {\bf active} and {\bf honest} device with $\delid{\device}$. Since $\device$ is honest, $\delid{\device}$ has been signed by the manufacture $\manu$.
Since $\device$ is an active device, it will have sent $\delid{\device}$ to some honest custodian $C$ by assumption.

Suppose $\CustodianConsensus$ does not abort. The honest custodian $C$ will validate $\delid{\device}$ and include it in its prospective list $\tilde \keydb_\epoch$. 
Since the procedure does not abort $\delid{\device}$ must be included in the final consensus list $\keydb_\epoch$.

Let $K$ be some entry in the final $\keydb_\epoch$. Since the procedure did not abort, the honest custodian $C$ must have validated the list. Therefore, $K$ has been signed by $\manu$ which concludes the proof.
\end{proof}

\vspace{0.5cm}

\lemmacustpriv*
\begin{proof}
No witness to the selection of $\delsel$ leaves the device until the hardware protocol {\bf $\bm{\device}$.\RevealChallenge$()$} is initiated. 
Furthermore, only $\mathsf{Enc}_{\pkj}(\delsel || \nonce)$ is returned.
Since there is at least one honest custodian, the adversary cannot possess every share needed to decrypt $\enc{\delsel || \nonce}$.
Therefore, $\adv$ cannot decrypt $\enc{\delsel || \nonce}$ with better than negligible probability by property of public key encryption schemes.

Additionally, by construction, each device selects $\delsel$ uniformly and independently at random.
By Assumption \ref{assumption:enclave}, this selection by the secure processor cannot be inferred. 
Therefore, $\adv$ cannot recover $\delsel$ with probability that is better than random guessing.
\end{proof}

\vspace{0.5cm}

\lemmaneedkeys*
\begin{proof}
Since Theorem \ref{thm:setup} holds for $\epoch$, the MerkleProof $\MerkleProof_\epoch$ is signed by the custodians and can be verified by $\device$ using the header $\keydbheader$ signed by the custodians. 
The custodian's signature has a root of trust originating from $\manu$'s signature during the initial custodian setup.
    
The only means to unlock $\device$ through \scheme is through the {\bf$\bm{\device}$.\DeviceUnlock$(\delselmal, \MerkleProof)$} protocol. 
Recall
    \[
    \delselmal = \{j||\delid{\device_j}||\sig_{\device_j}(\nonce) ~|~ j \in \delsel\}
    \] 
    and $\MerkleProof$ is a set of Merkle proofs verifying that $\vk_{\device_j}$ is at index $j$ in $\keydb$ for all $j \in \delsel$.
The proofs $\MerkleProof_\epoch$ can be verified via the Merkle hash in $H_\epoch$ that $\device$ has downloaded for epoch $\epoch$.
    
This input is required by the secure enclave in order to unlock. Suppose the adversary could generate this input \emph{without} access to all of the secret keys $\{ \sk_i ~|~ i \in \delsel\}$ belonging to the delegate devices. The adversary must then be able to either forge the signatures $\sig_i$ or fabricate the Merkle proofs $M$. We assume this only happens with negligible probability~\cite{merkle1987digital, cramer2000signature, hoffstein2001nss}.
\end{proof}

\vspace{0.5cm}

\lemmafrac*
\begin{proof}
Suppose the adversary corrupts some set $\mathcal{R}$ of devices with $|\mathcal{R}| = \pcorrupt$.
By construction, each non corrupted honest device  selects its delegation indices $\delsel$ such that $|\delsel| = \numdel$ uniformly at random and independently from all other devices.
Let $\threshdel = \numdel$. 
Since the indices are chosen uniformly at random no mater how the adversary selects $\mathcal{R}$, the probability that the adversary possess at least $\threshdel$ devices on $\device$'s delegation is bounded by
\begin{equation}
\label{eq:corrupt}
\left(\frac{\pcorrupt}{\numd}\right)^{\numdel}
    \end{equation}
    when $m \leq C$ and is $0$ when $m > C$.
We must therefore be able to select $\numdel$ such that $\left( \frac{\pcorrupt}{\numd} \right)^\numdel < \secprob$. 
In practice, the full probability is exponentially small with reasonable choices for $\numdel$ and $\threshdel$. 
See Section \ref{sec:practicalSecurity} for analysis of practical choices of parameters.
\end{proof}

\vspace{0.5cm}

\thmmain*
\begin{proof}
Let $\device$ be an honest uncorrupted device. Suppose $\adv$ corrupts some set of $\mathcal{R}$ devices with $|\mathcal{R}| = C$.

By Lemma \ref{lem:custpriv}, no matter which other devices the adversary has chosen to corrupt, the adversary cannot select a device on $\device$'s choice of delegation $\delsel$ with better than random chance. 
By Lemma \ref{lem:needkeys}, the adversary must corrupt all delegate devices in $\delsel$. 
By Lemma \ref{lem:frac} there exists some $\numdel$ such that adversary will not have corrupted all of $\delsel$ with probability less than $p$.

Now suppose $\adv$ is is given the decryption to $\chal$ after having selected $\mathcal{R}$.
With delegation size $\numdel$ the adversary needs to corrupt at least one additional honest device with probability $1-p$.
\end{proof}
% !TEX root = main.tex
\section{Additional Extensions}
\subsection{Incorporating Time-lock Mechanisms}
One of the goals outlined of the exceptional access mechanism described in~\cite{savage2018lawful} is to impose a time-lock mechanism to require a certain amount of time (e.g., 3 days~\cite{savage2018lawful}) before giving law enforcement the ability to proceed with the unlock protocol. 
In theory, such a mechanism would mitigate the effectiveness of a ``device swap'' attack (e.g., where the target is away from the phone for a weekend). 
However, upon careful consideration, we do not believe such a mechanism would be an effective deterrent for a dedicated adversary and would likely only be a cause of frustration to law enforcement in an urgent situation (e.g., terrorist attack).
However, we note that a time-lock mechanism can be added to \scheme as a modular extension. 
For example, using recent developments in Verifiable Delay Function~\cite{boneh2018survey, boneh2018verifiable}, law enforcement can provide a proof-of-time-elapse to the device when issuing the unlock token. 
Taking a different approach, we can use the ideas presented in~\cite{wright2018crypto}, to create time-lock puzzles which law enforcement must break to recover a part of the access token. 
Such an approach would both incur a monetary cost and a time-lock cost (as described in~\cite{wright2018crypto}).

\end{document}